\documentclass{llncs}

\usepackage{amssymb}
\usepackage{amsmath}

\newcommand{\sat}{\vDash}
\newcommand{\unsat}{\not \vDash}
\newcommand{\imp}{\rightarrow}
\newcommand{\R}[1]{\mu^{-1}_{#1}}
\newcommand{\Proj}[1]{\mu_{#1}}
\newcommand{\skt}{\text{skt}}
\newcommand{\apr}{ \Rightarrow_{APR}}
\newcommand{\I}[1]{\mathcal{I}_{#1}}

\newtheorem{defin}{Definition}
\newtheorem{lem}[theorem]{Lemma}
\newtheorem{cor}[theorem]{Corollary}

\newtheorem{exmp}{Example}
\newcommand{\vars}{\operatorname{vars}}
\newcommand{\shortrules}[6]{\noindent\begin{minipage}{#6ex}{\bfseries #1}\end{minipage} $\;$ #2 $\;\Rightarrow_{\text{#5}}\;$ #3 \par\smallskip\noindent #4} 

\title{First-Order Logic Theorem Proving and Model Building via Approximation and Instantiation}
\author{Andreas Teucke\inst{1,2} \and Christoph Weidenbach\inst{1}}

\institute{Max-Planck Institut for Informatics, Campus E1 4
66123 Saarbr\"ucken
Germany \and Graduate School of Computer Science, Saarbr\"ucken, Germany
}

\date{}

\begin{document}
\pagestyle{plain}
\maketitle

\begin{abstract}
In this paper we consider first-order logic theorem proving and model
building via approximation and instantiation. 
Given a clause set we propose its approximation into a simplified clause set where
satisfiability is decidable. The approximation extends the signature
and  preserves unsatisfiability: if the simplified clause set is satisfiable in some model, 
so is the original clause set in the same model interpreted in the original signature. 
A refutation generated by a decision procedure on the simplified clause set can 
then either be lifted to a refutation in 
the original clause set, or it guides 
a refinement excluding the previously found unliftable refutation. 
This way the approach is refutationally complete. We do not step-wise lift refutations
but conflicting cores, finite unsatisfiable clause sets representing at least one refutation.
The approach is dual to many
existing approaches in the literature because our approximation preserves unsatisfiability.
\end{abstract}

\section{Introduction} \label{sec:intro}

The Inst-Gen calculus by Ganzinger and Korovin~\cite{DBLP:conf/birthday/Korovin13} and its
implementation in iProver has shown to be very successfull. The calculus is based on a
under-approximation - instantiation refinement loop. A given first-order clause set is under-approximated
by finite grounding and afterwards a SAT-solver is used to test unsatisfiability. If the
ground clause set is unsatisfiable then a refutation for the original clause set is found.
If it is satisfiable, the model generated by the SAT-solver is typically not a model for
the original clause set. If it is not, it is used to instantiate the
original clause such that the found model is ruled out for the future.

In this paper we define a calculus that is dual to the Inst-Gen calculus.
A given first-order clause set is over-approximated into a decidable fragment
of first-order logic: a monadic, shallow, linear Horn (mslH) theory~\cite{Weidenbach99cade}.
If the over-approximated clause set is satisfiable, so is the original clause set.
If it is unsatisfiable, the found refutation is typically not a refutation for
the original clause set. If it is not, the refutation is analyzed to instantiate
the original clause set such that the found refutation is ruled out for the future.
The mslH fragment properly include first-order ground logic, but is also expressive
enough to represent minimal infinite models.

In addition to developing a new proof method for first-order logic this constitutes our second motivation
for studying the new calculus and the particular mslH approximation. It is meanwhile accepted that
a model-based guidence can significantly improve an automated reasoning calculus. The propositional
CDCL calculus~\cite{NieuwenhuisEtAl06} 
is one prominent example for this insight. In first-order logic, (partial) model operators
typically generate inductive models for which almost all interesting properties become undecidable,
in general. One way out of this problem is to generate a model for an approximated clause set, such
that important properties with respect to the original clause set are preserved. In the case of our calculus and approximation, a found model
can be effectively translated into a model for the original clause set. So our result is also a first
step towards model-based guidence in first-order logic automated reasoning.

For example, consider the first-order Horn clauses 
$S(x) \rightarrow P(x,g(x))$; $ S(a)$; $ S(b)$; $ S(g(x))$; $ \neg P(a, g(b))$; $ \neg P(g(x), g(g(x)))$
that are approximated (Section~\ref{sec:abstr}) into the mslH theory
$S(x), R(y) \rightarrow T(f_P(x,y))$; $ S(x) \rightarrow R(g(x))$; $ S(a)$; $ S(b)$; $S(g(x))$; $ \neg T(f_P(a, g(b)))$; $ \neg T(f_P(g(x), g(g(x))))$
where the relation $P$ is encoded by the function $f_P$ and the non-linear occurrence of $x$ in the first
clause is approximated by the introduction of the additional variable $y$. The approximated clause set
has two refutations: one using $\neg T(f_P(a, g(b)))$ and the second using $\neg T(f_P(g(x), g(g(x))))$ plus
the rest of the clauses, respectively. While the first refutation cannot be lifted, the second one is liftable to 
a refutation of the original clause set (Section~\ref{sec:lifting}). Actually, we do not consider refutations, but conflicting cores (Definition~\ref{conlfictCore}).
Conflicting cores are finite, unsatisfiable clause sets where variables are considered to be shared among clauses and rigid such that
any instantiation preserves unsatisfiability. Conflicting cores can be effectively generated out of refutations via instantiation
of (copies of) the input clauses involved in the refutation.
For the above second refutation the conflicting core of the approximated clause set is
$S(g(x)), R(g(g(x))) \rightarrow T(f_P(g(x),g(g(x))))$; $ S(g(x)) \rightarrow R(g(g(x)))$; $ S(g(x))$; $ \neg T(f_P(g(x), g(g(x))))$.\newline
In case the first refutation is selected for lifting, it fails, so the original clause set is refined (Section~\ref{abstrref}). The refinement
replaces the first clause with\newline
$S(a) \rightarrow P(a,g(a))$; $ S(b) \rightarrow P(b,g(b))$ and $ S(g(x)) \rightarrow P(g(x),g(g(x)))$.\newline
The approximation of the resulting new clause set does no longer enable a refutation using $\neg T(f_P(a, g(b)))$. Therefore, the refutation using $\neg T(f_P(g(x), g(g(x))))$ is found after refinement.
In case the original clause set contains a non-Horn clause, one positive literal is selected by the approximation.

The paper is now organized as follows. Section~\ref{sec:abstr} introduces some basic notions and the approximation relation $\apr$
that transforms any first-order clause set into an mslH theory.
The lifting of conflicting cores is described in Section~\ref{sec:lifting} and the respective abstraction refinement
in Section~\ref{abstrref} including soundness and completeness results. Missing proofs can be found in the appendix.
The paper ends with Section~\ref{sec:furewo} on future/related work and a conclusion.

\section{Linear Shallow Monadic Horn Approximation} \label{sec:abstr}

We consider a standard first-order language without equality where $\Sigma$ 
denotes the set of function symbols.
The symbols $x,y$ denote variables, $a,b$ constants, 
$f,g,h$ are functions and $s,t$ terms. 
Predicates are denoted by $S,P,Q,R$, literals by $E$, clauses by $C,D$, and sets of clauses by  $N,M$.
The term $t[s]_p$ denotes that the term $t$ has the subterm $s$ at position $p$. 
The notion is extended to atoms, clauses, and multiple positions.
A predicate with at most one argument is called monadic.
A literal is either an atom or an atom preceded by $\neg$ and it is then respectively  called positive or negative.
A term is shallow if it has at most depth one. 
It is called linear if there are no duplicate variable occurrences. 
A literal, where every term is shallow, is also called shallow. 
A clause is a multiset of literals which we write as an implication $\Gamma \imp \Delta$ where the atoms in $\Delta$ denote the positive literals
and the atoms in $\Gamma$ the negative literals.
If $\Gamma$ is empty we omit $\imp$, e.g., we write $P(x)$ instead of $\imp P(x)$ whereas if
$\Delta$ is empty $\imp$ is always shown.
If a clause has at most one positive literal, it is a Horn clause.
If there are no variables, then terms, atoms and clauses are respectively called ground.
A substitution $\sigma$ is a mapping from variables into terms denoted by pairs $\{x \mapsto t\}$.
If for some term (literal, clause) $t$, $t\sigma$ is ground, then $\sigma$ is a grounding substitution.


A Herbrand interpretation $I$ is a - possibly infinite - set of positive ground literals and $I$ is said to satisfy a 
clause $C= \Gamma \imp \Delta$, denoted by $I\sat C$, if  $\Delta\sigma \cap I \neq \emptyset$ or $\Gamma\sigma \not\subseteq I$ for every grounding substitution $\sigma$.
An interpretation $I$ is called a model of $N$ if $I$ satisfies $N$, $I\sat N$, i.e., 
$I\sat C$ for every $C\in N$. Models are considered \emph{minimal} with respect to set inclusion.
A set of clauses $N$ is satisfiable, if there exists a model that satisfies $N$. Otherwise the set is unsatisfiable.

\begin{defin}[Conflicting Core]\label{conlfictCore}
A finite clause set $N^\bot$ is a conflicting core if for all grounding substitutions $\tau$ the clause set $N^\bot\tau$ is unsatisfiable. 
$N^\bot$ is a conflicting core of $N$ if $N^\bot$ is a conflicting core 
and for every clause $C\in N^\bot$ there exists a $C'\in N$ such that $C=C'\sigma$.
\end{defin}

\begin{defin}[Specific Instances]\label{specInst}
Let $C$ be a clause and $\sigma_1$, $\sigma_2$ be two substitutions
such that $C\sigma_1$ and $C\sigma_2$ have no common instances.
Then the \emph{specific instances} of $C$ with
respect to $\sigma_1$, $\sigma_2$ are clauses $C\tau_1,\ldots,C\tau_n$ such that
(i)~any ground instance of $C$ is an instance of some $C\tau_i$,
(ii)~there is no $C\tau_i$ such that both $C\sigma_1$ and $C\sigma_2$ are 
instances of $C\tau_i$.
\end{defin}

The definition of specific instances can be extended to a single substitution $\sigma$.
In this case we require $C$ and $\sigma$ to be linear,
condition~(i) from Definition~\ref{specInst} above, $C\sigma = C\tau_1$ and no $C\tau_i$, 
$i\neq 1$ has a common instance with $C\tau_1$.
Note that under the above restrictions specific instances always exist~\cite{Lassez:1987:ERT:33031.33036}.

\begin{defin}[Approximation]
Given a clause set $N$ and a relation $\Rightarrow$ on clause sets with $N \Rightarrow N'$ then
(1)~$\Rightarrow$ is called an \emph{over-approximation}  if satisfiability of $N'$ implies satisfiability of $N$,
(2)~$\Rightarrow$ is called an \emph{under-approximation} if unsatisfiability of $N'$ implies unsatisfiability of $N$.
\end{defin}


Next we introduce our concrete over-approximation $\apr$ that 
eventually maps a clause set $N$ to an mslH
clause set $N'$. Starting from a clause set $N$ the transformation is parameterized
by a single monadic projection predicate $T$, fresh to $N$
and for each non-monadic predicate $P$ a projection function $f_P$ fresh to $N$. The approximation
always applies to a single clause and we establish on the fly an ancestor relation between 
the approximated clause(s) and the parent clause. The ancestor relation is needed for lifting and refinement.

\bigskip
\shortrules{Monadic}{$N\cup\{\Gamma \imp \Delta,P(t_1,\dots,t_n)\}$}{$N\cup\{\Gamma \imp \Delta,T(f_P(t_1,\dots,t_n))\}$}{provided $n>1$; $P(t_1,\dots,t_n)$ is the ancestor of $T(f_P(t_1,\dots,t_n))$}{MO}{10}

\bigskip
\shortrules{Horn}{$N\cup\{\Gamma \imp E_1, \dots,E_n\}$}{$N\cup\{\Gamma \imp E_i\}$}{provided $n>1$; $\Gamma \imp E_1, \dots,E_n$ is the ancestor of $\Gamma \imp E_i$}{HO}{10}

\bigskip
\shortrules{Shallow}{$N\cup\{\Gamma \imp E[s]_{p}\}$}{$N\cup\{S(x),\Gamma_1 \imp E[x]_{p}\}\cup\{ \Gamma_2 \imp S(s)\}$}{provided $s$ is a complex term, $p$ not a top position, $x$ and $S$ fresh, 
and $\Gamma_1 \cup \Gamma_2 = \Gamma$; $\Gamma \imp E[s]_{p}$ is the ancestor of $S(x),\Gamma_1 \imp E[x]_{p}$ and $\Gamma_2 \imp S(s)$}{SH}{10}

\bigskip
\shortrules{Linear}{$N\cup\{\Gamma \imp E[x]_{p,q}\}$}{$N\cup\{\Gamma\{x \mapsto x'\},\Gamma \imp E[x']_q\}$}{provided $x'$ is fresh, the positions $p$, $q$ denote
two different occurrences of $x$ in $E$; $\Gamma \imp E[x]_{p,q}$ is the ancestor of $\Gamma\{x \mapsto x'\},\Gamma \imp E[x']_q$}{LI}{10}

\bigskip
For the Horn transformation, the choice of the $E_i$ is arbitrary.
In the Shallow rule, $\Gamma_1$ and $\Gamma_2$ can be arbitrarily chosen as long as they ``add up'' to $\Gamma$.
The goal, however, is to minimize the set of common variables 
$\vars(\Gamma_2,s) \cap \vars(\Gamma_1,E[x]_p)$.
If this set is empty the Shallow transformation is satisfiability preserving.
In rule Linear, the duplication of $\Gamma$ is not needed if $x\not\in\vars(\Gamma)$.

\begin{defin}[$\apr$] \label{def:approx:trans}
The overall approximation $\apr$ is given by $\apr\; = \;\Rightarrow_{\text{MO}}\cup\Rightarrow_{\text{HO}}\cup\Rightarrow_{\text{SH}}\cup\Rightarrow_{\text{LI}}$
with a preference on the different rules where  Monadic precede Horn precede Shallow precede Linear transformations.\\
\end{defin}

\begin{defin}
Given a non-monadic n-ary predicate $P$, projection predicate $T$, and projection function $f_P$,
define the injective function $\Proj{P}(P(t_1,\dots,t_n)) := T(f_p(t_1,\dots, t_n))$
and $\Proj{P}(Q(s_1,\dots,s_m)) := Q(s_1,\dots,s_m)$ for any atom with a predicate symbol different from $P$.
The function is extended to clauses, clause sets and interpretations.
\end{defin}

\begin{lem}[$\apr$ is sound and terminating]\label{Apr_sound}
The approximation rules are sound and terminating:
(i)~$\apr$  terminates
(ii)~the Monadic transformation is an over-approximation
(iii)~the Horn transformation is an over-approximation
(iv)~the Shallow transformation is an over-approximation
(v)~the Linear transformation is an over-approximation
\end{lem}
\begin{proof}
(i)~The transformations can be considered sequentially, because of the imposed rule preference (Definition~\ref{def:approx:trans}). 
The monadic transformation strictly reduces the number of non-monadic atoms.
The Horn transformation strictly reduces the number of non-Horn clauses.
The Shallow transformation strictly reduces the multiset of term depths of the newly introduced clauses compared
to the removed ancestor clause.
The linear transformation strictly reduces the number of duplicate variables occurrences in positive literals. 
Hence $\apr$  terminates.

(ii) Consider a transformation $N_k\Rightarrow^*_{\text{MO}} N_{k+j}$ that exactly removes all occurrences of atoms $P(t_1,\dots,t_n)$ and
replaces those by atoms $T(f_P(t_1,\dots,t_n))$.
Then, $N_{k+j}=\Proj{P}(N_k) $ and $N_k=\R{P}(N_{k+j}) $.
Let $I$ be a model of $N_{k+j}$ and $C \in N_{k}$.  
Since $ \Proj{P}(C) \in N_{k+j}$ , $I \sat \Proj{P}(C)$ and thus, $\R{P}(I)\sat C$.
Hence, $\R{P}(I)$ is a  model of $N_k$.
Therefore, the Monadic transformation is an over-approximation.

(iii) Let $N\cup\{\Gamma \imp E_1, \dots,E_n\} \Rightarrow_{\text{HO}} N\cup\{\Gamma \imp E_i\}$.
The clause $\Gamma \imp E_i$ subsumes the clause $\Gamma \imp E_1, \dots,E_n$. Therefore, for any $I$
if $I\models \Gamma \imp E_i$  then $I\models \Gamma \imp E_1, \dots,E_n$.
Therefore, the Horn transformation is an over-approximation.

(iv) Let $N_k = N\cup\{\Gamma \imp E[s]_{p}\} \Rightarrow_{\text{SH}} N_{k+1}= N\cup\{S(x),\Gamma_1 \imp E[x]_{p}\}\cup\{ \Gamma_2 \imp S(s)\}$.
Let $I$ be a model of $N_{k+1}$ and $C\in N_k$ be a ground clause. 
If $C$ is an instance of a clause in $N$, then $I \models C$.
Otherwise $C = (\Gamma \imp E[s]_{p})\sigma$ for some ground substitution $\sigma$.
Then $S(s)\sigma,\Gamma_1\sigma \imp E[s]_{p}\sigma = (S(x),\Gamma_1 \imp E[x]_{p})\{x \mapsto s\}\sigma \in N_{k+1}$ 
and $ \Gamma_2\sigma \imp S(s)\sigma = (\Gamma_2 \imp S(s))\sigma \in N_{k+1}$.
Since $I \models N_{k+1} $, $I$ also satisfies the resolvent $\Gamma_1\sigma,\Gamma_2\sigma \imp E[s]\sigma = C$. 
Hence $I \models N_k$.
Therefore, the Shallow transformation is an over-approximation.

(v)  Let $N_k =N\cup\{\Gamma \imp E[x]_{p,q}\} \Rightarrow_{\text{LI}} N_{k+1}=N\cup\{\Gamma\{x \mapsto x'\},\Gamma \imp E[x']_q\}$.
Let $I$ be a model of $N_{k+1}$ and $C\in N_k$ be a ground clause. 
If $C$ is an instance of a clause in $N$, then $I \models C$.
Otherwise $C = (\Gamma \imp E[x]_{p,q})\sigma$ for some ground substitution $\sigma$.
Then $(\Gamma\{x \mapsto x'\},\Gamma \imp E[x']_q)\{x' \mapsto x\}\sigma \in N_{k+1}$ and
$I \models (\Gamma\{x \mapsto x'\},\Gamma \imp E[x']_q)\{x' \mapsto x\}\sigma = (\Gamma,\Gamma \imp E[x]_q)\sigma \models C $.
Hence $I \models N_k$.
Therefore, Linear transformation is an over-approximation.
\end{proof}

\begin{cor}
(i)~$\apr$ is an over-approximation.
(ii)~If $N \apr^* N'$, $P_1,\dots,P_n$ are the non-monadic predicates in $N$ and $N'$ is satisfied by model $I$, \\
then  $\R{P_1}(...(\R{P_n}(I)))$  is a model of $N$.
\end{cor}
\begin{proof}
Follows from Lemma~\ref{Apr_sound} (ii)-(v).
\end{proof}

In addition to being an over-approximation, the minimal model (with respect to set inclusion)
 of the eventual approximation
generated by  $\apr$ preserves the skeleton term structure
of the original clause set, if it exists. The refinement introduced in Section~\ref{abstrref}
instantiates clauses. Thus it contributes to finding a model or a refutation.

\begin{defin}[Term Skeleton]\label{termSkel}
The term skeleton of term $t$ , $\skt(t)$, is defined as \\
(1) $\skt(x)=x'$, where $x'$ is a fresh variable \\
(2) $\skt(f(s_1,\dots,s_n))= f(\skt(s_1),\dots,\skt(s_n))$. 
\end{defin}

\begin{lem}\label{termSkelLem2}
Let $N_k$ be a monadic clause set and $N_0$ be its approximation via  $\apr$. 
Let $N_0$ be satisfiable and $I$ be a minimal model for $N_0$.
If $P(s) \in I$  and $P$ is a predicate in $N_k$, then there exists a clause 
$C = \Gamma \imp \Delta,P(t) \in N_k$ and a substitution $\sigma$ such that
$s = \skt(t)\sigma$ and for each variable $x$ and predicate $S$ with $C = S(x),\Gamma' \imp \Delta,P(t[x]_p)$, $S(s'') \in I$, where $s=s[s'']_p$.
\end{lem}

\begin{proof}
By induction on $k$.\\
For the base $N_k=N_0$, assume there is no $C \in N_0$ with $C\sigma= \Gamma \imp \Delta,P(s)$ and $\Gamma \subseteq I$.
Then $I\setminus \{P(s)\}$ is still a model of $N_0$ and therefore $I$ is not minimal.\\
Let  $N = N_k\apr N_{k-1} \apr^* N_0$, $P(s) \in  I$ and $P$ is a predicate in $N_k$ and hence also in $N_{k-1}$. 
By the induction hypothesis, there exist a clause $C = \Gamma \imp \Delta,P(t) \in N_{k-1}$ and a substitution $\sigma$ such that
$s = \skt(t)\sigma$ and for each variable $x$ and predicate $S$ with $C = S(x),\Gamma' \imp \Delta,P(t[x]_p)$, $S(s'') \in I$, where $s=s[s'']_p$.
The first approximation rule application is either a Linear, a Shallow or a Horn transformation,
considered below by case analysis.

Horn Case. Let $\apr$ be a Horn transformation that replaces $\Gamma'' \imp \Delta',Q(t')$ with $\Gamma'' \imp Q(t')$.
If $C \neq \Gamma'' \imp Q(t')$, then $C \in N_k$ fulfills the claim.
Otherwise, $\Gamma'' \imp \Delta',Q(t) \in N_k$ fulfills the claim since $P=Q$ and $\Gamma' = \Gamma''$.

Linear Case. Let $\apr$ be a linear transformation that replaces $C_k = \Gamma'' \imp E[x]_{p,q}$ with $C_{k-1}=\Gamma'',\Gamma''\{x \mapsto x'\} \imp E[x']_{q}$.
If $C \neq C_{k-1}$, then $C \in N_k$ fulfills the claim.
Otherwise, $C_k=\Gamma'' \imp P(t)\{x' \mapsto x\} \in N_k$ fulfills the claim since $s = \skt(t)\sigma = \skt(t\{x' \mapsto x\})\sigma$ and $\Gamma'' \subseteq \Gamma'',\Gamma''\{x \mapsto x'\}$.

Shallow Case. Let $\apr$ be a shallow transformation that replaces $C_k = \Gamma'' \imp E[s']_{p}$ with $C_{k-1}=S(x),\Gamma_1 \imp E[x]_p$ and $C'_{k-1}=\Gamma_2 \imp S(s')$.
Since $S$ is fresh, $C \neq C'_{k-1}$.
If $C \neq C_{k-1}$, then $C \in N_k$ fulfills the claim.
Otherwise, $C = C_{k-1} = S(x),\Gamma_1 \imp P(t[x]_p)$ and hence, $s = \skt(t[x]_p)\sigma$ and $S(s'')\in I$ for $s=s[s'']_p$.
Then by the induction hypothesis, there exist a clause $C_S = \Gamma_S \imp \Delta_S,S(t_S) \in N_{k-1}$ and a substitution $\sigma_S$ such that
$s'' = \skt(t_S)\sigma_S$ and for each variable $x$ and predicate $S'$ with $C_S = S'(x),\Gamma'_S \imp \Delta_S,P(t_S[x]_q)$, $S'(s''') \in I$, where $s''=s''[s''']_q$.
By construction, $C_S = C'_{k-1}$. 
Thus,  $s''= \skt(s')\sigma_S$ and $s = \skt(t[x]_p)\sigma$ imply there exists a $\sigma''$ such that $s =  \skt(t[s']_p)\sigma"$.
Furthermore since $\Gamma_1 \cup \Gamma_2 = \Gamma''$, if $C_k= S'(x),\Gamma''' \imp P(t[s']_{p})[x]_q$, then 
either $S'(x) \in \Gamma_1$ and thus  $S'(s'''') \in I$, where $s=s[s'''']_q$, or 
$S'(x) \in \Gamma_2$ and thus  $S'(s'''') \in I$, where $s[s'']_p=(s[s'']_p)[s'''']_q$.
Hence, $C_k \in N_k$  fulfills the claim.
\end{proof}

\begin{lem}\label{termSkelLem}
Let $N$ be a clause set and $N'$ be its approximation via  $\apr$. Let $N'$ be
satisfiable and $I$ be a minimal model for $N'$.
If $P(s) \in I$ $(T(f_p(s_1,\ldots,s_n))\in I)$ 
and $P$ is a predicate in $N$, then there exist a clause 
$\Gamma \imp \Delta,P(t) \in N$  ($\Gamma \imp \Delta,P(t_1,\ldots,t_n) \in N$)
and a substitution $\sigma$ such that
$s = \skt(t)\sigma$ ($s_i = \skt(t_i)\sigma$ for all $i$).
\end{lem}

\begin{proof}
Let $P_1,\ldots,P_n$ be the non-monadic predicates in $N$ and $N_{MO}=\Proj{P_1}(\ldots(\Proj{P_n}(N)))$. Then, $N_{MO}$ is monadic and also has $N'$ as its approximation via $\apr$.

 Let $P(s) \in I$ and $P$ is a predicate in $N$. Since $P$ is monadic, $P$ is a predicate in $N_{MO}$.
Hence by Lemma ~\ref{termSkelLem2}, there exists a clause $\Gamma \imp \Delta,P(t) \in N_{MO}$ 
and a substitution $\sigma$ such that $s = \skt(t)\sigma$. 
Then, $\R{P_1}(\ldots(\R{P_n}(\Gamma \imp \Delta,P(t)))\ldots)= \R{P_1}(\ldots(\R{P_n}(\Gamma)\ldots) \imp \R{P_1}(\ldots(\R{P_n}(\Delta)\ldots),P(t)\in N$ fulfills the claim.

Let $T(f_p(s_1,\ldots,s_n))\in I$ and $P$ is a predicate in $N$. $T$ is monadic and a predicate in $N_{MO}$.
Hence by Lemma ~\ref{termSkelLem2}, there exists a clause $\Gamma \imp \Delta,T(t) \in N_{MO}$ 
and a substitution $\sigma$ such that $f_p(s_1,\ldots,s_n) = \skt(t)\sigma$.
Therefore, $t=f_p(t_1,\ldots,t_n)$ with  $s_i = \skt(t_i)\sigma$ for all $i$.
Then, $\R{P_1}(\ldots(\R{P_n}(\Gamma \imp \Delta,T(f_p(t_1,\ldots,t_n))))\ldots)=\R{P_1}(\ldots(\R{P_n}(\Gamma)\ldots) \imp \R{P_1}(\ldots(\R{P_n}(\Delta)\ldots),P(t_1,\ldots,t_n) \in N$ fulfills the claim.
\end{proof}

The above lemma also holds if satisfiability of $N'$ is dropped and $I$ is replaced
by the superposition partial minimal model operator~\cite{Weidenbach01handbook}.

\section{Lifting the Conflicting Core} \label{sec:lifting}

Given a monadic, linear, shallow, Horn approximation $N_k$ of $N$ 
and a conflicting core $N^\bot_k$ of $N_k$, using the transformations
provided in this section we attempt to lift $N^\bot_k$ to a conflicting
core $N^\bot$ of $N$.
In case of success this shows the unsatisfiability of $N$.
In case an approximation step cannot be lifted the original clause
set is refined by instantiation, explained in the next section.

Let $N_k$ be an unsatisfiable monadic, linear, shallow, Horn approximation. 
Since $N_k$ belongs to a decidable first-order fragment,
we expect an appropriate decision procedure to generate a proof of unsatisfiability for $N_k$, 
e.g., ordered resolution with selection \cite{Weidenbach99cade}.
A conflicting core can be straightforwardly generated out of a resolution refutation 
by applying the substitutions of the proof to the used input clauses.

Starting with a resolution refutation, in order to construct the conflicting core, 
we begin with the singleton set containing the pair of empty clause and the empty substitution.
Furthermore, we assume that all input clauses from $N_k$ used in the refutation are
variable disjoint. 
Then we recursively choose a pair $(C,\sigma)$ from the set where $C \notin N_k$.
There exists a step in the refutation that generated this clause. 
In the case of a resolution inference, there are two parent 
clauses $C_1$ and $C_2$ in the refutation and two substitutions $\sigma_1$ and $\sigma_2$
such that $C$ is the resolvent of $C_1\sigma_1$ and $C_2\sigma_2$.
In the case of a factoring inference, there is one parent 
clause $C'$ in the refutation and a substitution $\sigma'$
such that $C$ is the factor of $C'\sigma'$.  
Replace  $(C,\sigma)$ by $(C_1,\sigma_1\sigma)$ and $(C_2,\sigma_2\sigma)$ or 
by $(C',\sigma'\sigma)$ respectively.
The procedure terminates in linear time in the size of the refutation.
For each pair $(C,\sigma)$, collect the clause $C\sigma$, resulting
in a conflicting core $N^\bot_k$ of $N_k$.

\begin{exmp}\label{coreex}
Let $N = \{ P(x,x');$ $P(y,a),P(z,b) \imp \}$ with signature $\Sigma = a/0, b/0$. 
$N$ is unsatisfiable and  a possible resolution refutation is resolving $P(b,a)$ 
and $P(a,b)$ with $P(b,a),P(a,b)\imp$. From this we get the conflicting core \\
$N^\bot_{ba}=\{ P(b,a);P(a,b);P(b,a),P(a,b)\imp\}$.

An alternative refutation is to resolve  $P(x,x')$ and $P(y,a),P(z,b) \imp$ with 
substitution $\{x\mapsto y;x' \mapsto a\}$ 
and then the resolvent and $P(x,x')$ with substitution $\{x\mapsto z;x' \mapsto b\}$. 
From this refutation we construct the conflicting core  
$N^\bot_{yz}=\{ P(y,a);$ $P(z,b);$ $P(y,a),P(z,b)\imp\}$.

\end{exmp} 

Note that in Example \ref{coreex}  $N^\bot_{yz}$ is more general than $N^\bot_{ba}$ since $N^\bot_{yz}\{y \mapsto b; z \mapsto a\} = N^\bot_{ba}$.
A conflicting core is minimal in that it represents the most general clauses corresponding to the refutation from that it is generated.

\subsubsection{Lifting the Monadic Transformation.}
Since the Monadic transformation is satisfiability preserving, lifting always succeeds by replacing any $T(f_P(t_1,\dots,t_n))$ atoms in the core  by $P(t_1,\dots,t_n)$.

\begin{exmp}\label{LiftMonadicEx}
Let $N_0 = \{ P(x,x');$ $P(y,a),P(z,b) \imp \}$. 
Then $N_k=  \{ T(f_P(x,x'));$ $T(f_P(y,a)),T(f_P(z,b)) \imp \}$ is a Monadic transformation of $N_0$
and a conflicting core is $N^\bot_k= \{ T(f_P(y,a));$ $T(f_P(z,b));$ $T(f_P(y,a)),T(f_P(z,b)) \imp \}$.
Reverting the atoms in  $N^\bot_k$ gives $N^\bot=\{ P(y,a);$ $P(z,b);$ $P(y,a),P(z,b)\imp\}$ a conflicting core of $N_0$.
\end{exmp} 


\begin{lem}[Lifting the Monadic Transformation]\label{liftmonadic}
Let $N_k\Rightarrow^*_{\text{MO}} N_{k+l}$ be the transformation  that exactly removes all occurrences of atoms $P(t_1,\dots,t_n)$ and replaces those by atoms $T(f_P(t_1,\dots,t_n))$.
If $N^\bot_{k+l}$ is a conflicting core for $N_{k+l}$ then there is a conflicting core $N^\bot_k$ of $N_k$.
\end{lem}

\begin{proof}
Since the Monadic transformation is satisfiability preserving, unsatisfiability of $N_{k+l}$ directly implies unsatisfiability of $N_{k}$ and the existence of a conflicting core of $N_{k}$.
\end{proof}

\subsubsection{Lifting the Horn Transformation.}
For a Horn transformation there are two ways for lifting.
The first, directly lifting the core, only succeeds in special cases, where the original clause and its approximation are  equivalent 
for the instantiations appearing in the core.

\begin{exmp}\label{LiftHornEx}
Let $N_0 = \{ P(a,b) \imp;$ $ P(x,b),P(a,y) \}$. 
Then $N_k=  \{ P(a,b) \imp;$ $ P(x,b) \}$ is a Horn transformation of $N_0$
and a conflicting core is $N^\bot_k= \{ P(a,b) \imp;$ $ P(a,b) \}$.
By substituting $y$ with $b$,  $N^\bot_k$ lifts to $N^\bot=\{  P(a,b) \imp;$ $ P(a,b) , P(a,b) \}$ a conflicting core of $N_0$.
\end{exmp} 

\begin{lem}[Lifting the Horn Transformation (direct)]\label{lifthorn1}
Let $N_k\Rightarrow_{\text{HO}} N_{k+1}$ where 
$N_k = N\cup\{\Gamma \imp E_1, \dots,E_n\}$ and $N_{k+1} = N\cup\{\Gamma \imp E_i\}$.
Let $N^\bot_{k+1}$ be a conflicting core of $N_{k+1}$.
If for all $(\Gamma \imp E_i)\sigma_j\in N^\bot_{k+1}$, $1\leq j\leq m$ there is a 
substitution $\sigma'_j$ such that
$N_k^j\tau_j\models  (\Gamma \imp E_1, \dots,E_n)\sigma'_j\rightarrow (\Gamma \imp E_i)\sigma_j$,
such that $N_k^j\subseteq N_k$ and 
$N_k^j\tau_j \cup \{(\Gamma \imp E_1, \dots,E_n)\sigma'_j, \neg(\Gamma \imp E_i)\sigma_j\}$
is a conflicting core, then
$N^\bot_{k+1}\setminus \{(\Gamma \imp E_i)\sigma_j \mid 1\leq j\leq m\} \cup \{(\Gamma \imp E_1, \dots,E_n)\sigma'_j\mid 1\leq j\leq m\}\cup \bigcup\limits_j N_k^j\tau_j$ is
a conflicting core of $N_k$.
\end{lem}

\begin{proof}
Let  $\sigma$ be a grounding substitution for $N^\bot_{k}$ and $N^\bot_{k+1}$. 
Since $N_k\models  (\Gamma \imp E_1, \dots,E_n)\sigma'_j\rightarrow (\Gamma \imp E_i)\sigma_j$,
$N^\bot_{k}\sigma \models N^\bot_{k}\sigma \cup \{(\Gamma \imp E_i)\sigma_j \mid 1\leq j\leq m\}\sigma \models N^\bot_{k+1}\sigma$.
Hence, $N^\bot_{k}\sigma$ is unsatisfiable because $N^\bot_{k+1}\sigma$ is unsatisfiable.
Therefore, $N^\bot_{k}$ is an conflicting core of $N_k$.
\end{proof}

Of course, the condition $N_k^j\tau_j\models  (\Gamma \imp E_1, \dots,E_n)\sigma'_j\rightarrow (\Gamma \imp E_i)\sigma_j$ itself is undecidable, in general. The above lemma is meant to be a justification for
the cases where this relation can be decided, e.g, by reduction. In general, the next lemma applies.
We assume any non-Horn clauses have exactly two positive literals. 
Otherwise, we would have first redefined pairs of positive literals using fresh predicates. 
Further assume w.l.o.g. that Horn transformation always chooses the first positive Literal of a non-Horn clause.

The indirect method uses the information from the conflicting core to replace the non-Horn clause with a satisfiable equivalent unit clause,
which is then  solved recursively. Since this unit clause is already Horn, we lifted
one Horn approximation step.

\begin{exmp}
Let $N_k = \{ P(a),Q(a);  P(x)\imp\}$. 
The Horn transformation $N_k=\{  P(a); P(x) \imp\}$ has a conflicting core 
$N^\bot_k= \{  P(a); P(a) \imp\}$.
$N^\bot_k$ abstracts a resolution refutation with $\bot$ as the result.
If we replace  $ P(a)$ with $P(a),Q(a)$ in such a refutation,  the result will be $Q(a)$ instead and hence $N_k \sat Q(a)$ 
Since $Q(a)$ subsumes $P(a),Q(a)$, \\ $N_k$ is satisfiable if $N'_k=\{ Q(a);  P(x)\imp\}$ is too.
\end{exmp}

\begin{lem}[Lifting the Horn Transformation (indirect)]\label{lifthorn2}
Let $N$ be a set of variable disjoint clauses, 
$N\apr^* N_k \Rightarrow_{\text{HO}} N_{k+1}$,
$N_k = N\cup\{\Gamma \imp E_1,E_2\}$ and $N_{k+1} = N\cup\{\Gamma \imp E_1\}$ and
$N^\bot_{k+1}$ be a conflicting core of $N_{k+1}$ where Lemma~\ref{lifthorn1} does not apply.
Let $(\Gamma \imp E_1)\sigma \in N^\bot_{k+1}$ 
, where $\sigma$ is a variable renaming and
$N_k^j\tau_j\not \models  (\Gamma \imp E_1,E_2)\sigma'_j\rightarrow (\Gamma \imp E_1)\sigma$ for any $N_k^j\subseteq N_k$,$\tau_j$ and $\sigma'_j$. 
If there exists a conflicting core $N^\bot$ of $N\cup\{E_2\}$, then a  conflicting core of $N_k$ exists.
\end{lem}

\begin{proof}
From the conflicting core $N^\bot_{k+1}$, we can conclude that there exists an unsatisfiability proof of $N_{k+1}$ which derives $\bot$ and uses  $(\Gamma \imp E_1)\sigma$ as the only instance of $\Gamma \imp E_1$. If we were to replace $(\Gamma \imp E_1)\sigma$ by  $(\Gamma \imp E_1,E_2)\sigma$, the unsatisfiability proof's root clause would instead be $E_2\sigma$.
Hence, we know that $N_k\models N_k \cup \{E_2\sigma\}$. 
Furthermore, $N_k\models  N \cup \{E_2\sigma\}$ since $E_2\sigma$ subsumes $\Gamma \imp E_1,E_2$.

Let $E_2\sigma_j \in N^\bot$ for $1 \leq j \leq m$ and $N^{E_2}_k= N^\bot_{k+1} \setminus \{(\Gamma \imp E_1)\sigma\} \cup \{(\Gamma \imp E_1,E_2)\sigma\}$
Then  $N^\bot \setminus \{E_2\sigma_j \mid 1 \leq j \leq m\} \bigcup\limits_j N^{E_2}_{k}\sigma_j$ is a conflict core of  $N_k$.
\end{proof}

Note that $N_{k}$ now again contains the Non-Horn clause $\Gamma \imp E_1,E_2$. 
Then, in a following indirect Horn lifting step $\Gamma \imp E_1,E_2$ can not necessarily be again replaced by  $E_2\sigma$.
Hence, the indirect Horn lifting needs to be repeated.

\subsubsection{Lifting the Shallow Transformation.}
A Shallow transformation introduces a new predicate $S$, which is removed in the lifting step.
We take all clauses with $S$-atoms in the conflicting core and generate any 
possible resolutions on $S$-atoms. 
The resolvents, which don't contain $S$-atoms anymore, then replace their parent clauses in the core.
Lifting succeeds if all introduced resolvents are instances of clauses before the shallow
transformation.

\begin{exmp}\label{LiftShallowEx}
Let $N_0 = \{ P(x), Q(y) \imp R(x,f(y)); P(a);Q(b); R(a,f(b))\imp \}$.
Then $N_k = \{ S(x'),P(x) \imp R(x,x'); Q(y) \imp S(f(y)); P(a);  Q(b); $ $ R(a,f(b)) \imp \}$ is a Shallow transformation of $N_0$
and  a conflicting core is $N^\bot_k= S(f(b)),$ $P(a) \imp R(a,f(b));Q(b) \imp S(f(b));P(a);  Q(b);  R(a,f(b)) \imp$.
By replacing \\
$S(f(b)),P(a) \imp R(a,f(b))$ and $Q(b) \imp S(f(b))$ with the resolvent, $N^\bot_k$ lifts to $N^\bot =\{ P(a), Q(b) \imp R(a,f(b)); P(a);Q(b);  R(a,f(b)) \imp\}$ a conflicting core of $N_0$.
\end{exmp} 

\begin{lem}[Lifting the Shallow Transformation]\label{liftshallow}
Let $N_k\Rightarrow_{\text{SH}} N_{k+1}$ where $N_k = N\cup\{\Gamma \imp E[s]_{p}\}$ and $N_{k+1} = N\cup\{S(x),\Gamma_1 \imp E[x]_{p}\}\cup\{ \Gamma_2 \imp S(s)\}$.
Let $N^\bot_{k+1}$ be a conflicting core of $N_{k+1}$. 
Let $N_S$ be the set of all resolvents from  clauses from $N^\bot_{k+1}$ on the $S$ literal.
If for all clauses $C_j\in N_S$, $1\leq j\leq m$ there is a substitution $\sigma_j$ such that $C_j = (\Gamma \imp E[s]_{p})\sigma_j$ then
$N^\bot_{k+1}\setminus\{C\mid C\in N^\bot_{k+1}\text{ and contains an }S\text{-atom}\}\cup\{(\Gamma \imp E[s]_{p})\sigma_j\mid 1\leq j\leq m\}$ is a conflicting core of $N_k$.
\end{lem}

\begin{proof}
Let  $\sigma$ be a grounding substitution for $N^\bot_{k}$ and $N^\bot_{k+1}$ and $I$ be an interpretation. 
As  $N^\bot_{k+1}\sigma$ is unsatisfiable, there is a clause $D\in N^\bot_{k+1}\sigma$ such that $I\unsat D$. \\
If $D$ does not contain an $S$-atom, then $D \in N^\bot_{k}\sigma$ and hence $I\unsat N^\bot_{k}\sigma$.\\
Now assume only clauses that contain $S$-atoms are false under $I$. 
By construction, any such clause is equal to either  $(S(x),\Gamma_1 \imp E[x]_{p})\sigma'=C_1\sigma'$ or $(\Gamma_2 \imp S(s))\sigma'=C_2\sigma'$ for some substitution $\sigma'$.
Let $I' :=  \{S(s)\sigma' \mid $ $C_2\sigma' \in N^\bot_{k+1}\sigma \mathrm{ ~and~ } I\unsat C_2\sigma' \} \cup I \setminus \{ S(x)\sigma' ~\vert~  C_1\sigma' \in N^\bot_{k+1}\sigma \mathrm{ ~and~ } I\unsat C_1\sigma' \}$, i.e.,
we change the truth value for $S$-Literals such that the clauses unsatisfied under $I$ are satisfied under $I'$.\\
Since  $I$ and $I'$ only differ on literals with predicate $S$ and $N^\bot_{k+1}\sigma$ is unsatisfiable, some clause $C$, containing an $S$-atom and satisfied under $I$, has to be false under $I'$.\\
Let $C=C_1\sigma_1$. Since $I\sat C$, $S(x)\sigma_1$ was added to $I'$ by some clause $D =C_2\sigma_2$, where $S(s)\sigma_2 = S(x)\sigma_1$.
Hence, $C$ and $D$ can be resolved on their $S$-literals and the resolvent $R$ is in $N^\bot_{k}\sigma$.
Since $I\unsat D$, $I'\unsat C$ and $R$ contains no $S$-atom, $I\unsat R$ and therefore $I\unsat N^\bot_{k}\sigma$.\\
For $C=C_2\sigma_2$ the proof is analogous.\\
Thus, for all interpretations $I$ and grounding substitutions $\sigma$, $I\unsat N^\bot_{k}\sigma$ and hence $N^\bot_{k}\sigma$ is a conflicting core of $N_k$.
 \end{proof}

\subsubsection{Lifting the Linear Transformation.}
In order to lift a Linear transformation the remaining and the newly introduced variable  
need to be instantiated the same term. 

\begin{exmp}\label{liftLinEx}
Let $N_{k-1} = \{ P(x,x);$ $P(y,a),P(z,b) \imp \}$. 
Then $N_k=  \{ P(x,x');$ $P(y,a),P(z,b) \imp \}$ is a Linear transformation of $N_{k-1} $ and
and  $N^\bot_k=\{ P(a,a);$ $P(b,b);$ $P(a,a),P(b,b)\imp\}$ is a conflicting core of $N_k$.
Since $P(a,a)$ and $P(b,b)$ are instances of $P(x,x)$ lifting succeeds and $N^\bot_k$ is also a core of $N_{k-1} $.
\end{exmp} 

\begin{lem}[Lifting the Linear Transformation]\label{liftlin}
Let $N_k\Rightarrow_{\text{LI}} N_{k+1}$ where $N_k = N\cup\{\Gamma \imp E[x]_{p,q}\}$ and $N_{k+1} = N\cup\{\Gamma\{x \mapsto x'\},\Gamma \imp E[x']_q\}$.
Let $N^\bot_{k+1}$ be a conflicting core of $N_{k+1}$.
If for all $(\Gamma\{x \mapsto x'\},\Gamma \imp E[x']_q)\sigma_j\in N^\bot_{k+1}$, $1\leq j\leq m$ we have $x\sigma_j = x'\sigma_j$ then
$N^\bot_{k+1}\setminus\{(\Gamma\{x \mapsto x'\},\Gamma \imp E[x']_q)\sigma_j \mid 1\leq j\leq m\}\cup\{(\Gamma \imp E[x]_{p,q})\sigma_j\mid 1\leq j\leq m\}$ is a conflicting core of $N_k$.
\end{lem}

\begin{proof}
Let $\sigma$ be a grounding substitution for $N^\bot_{k}$ and $N^\bot_{k+1}$. 
As $x\sigma_j = x'\sigma_j$ for $1\leq j\leq m$, $(\Gamma \imp E[x]_{p,q})\sigma_j\sigma \sat (\Gamma,\Gamma \imp E[x]_{p,q})\sigma_j\sigma = (\Gamma\{x \mapsto x'\},\Gamma \imp E[x']_q)\sigma_j\sigma $.
Hence, $N^\bot_{k}\sigma \sat N^\bot_{k}\sigma \cup \{(\Gamma\{x \mapsto x'\},\Gamma \imp E[x']_q)\sigma_j\sigma \mid 1\leq j\leq m\} \sat N^\bot_{k+1}\sigma$.
Since $N^\bot_{k+1}\sigma$ is unsatisfiable $N^\bot_{k}\sigma$ is unsatisfiable as well.
Therefore,  $N^\bot_{k}$ is a conflicting core of $N_k$.
\end{proof}

\subsubsection{Lifting with Instantiation.}
By definition, if $N^\bot$ is a conflicting core of $N$, then $N^\bot\tau$ is also a conflicting core of $N$ for any $\tau$. 
Example \ref{liftex} shows it is sometimes possible to instantiate a conflicting core, where no lifting lemma applies, into a core, where one does.
This then still implies a successful lifting.

\begin{exmp}\label{liftex}
Let $N_{k-1} = \{ P(x,x);$ $P(y,a),P(z,b) \imp \}$. 
Then $N_k=  \{ P(x,x');$ $P(y,a),P(z,b) \imp \}$ is a Linear transformation of $N_{k-1} $ and
and  $N^\bot_k=\{ P(y,a);$ $P(b,b);$ $P(y,a),P(b,b)\imp\}$ is a conflicting core of $N_k$.
Since for $P(y,a)= P(x,x')\sigma$ $x\sigma=y\neq a=x'\sigma$ Lemma \ref{liftlin} is not applicable.

However,  Lemma \ref{liftlin} can be applied on $N^\bot_k\{y \mapsto a; z \mapsto b\}= \{P(a,a);$ $P(b,b);$ $P(a,a),P(b,b)\imp\}$.
\end{exmp}

\section{Approximation Refinement}\label{abstrref}

In the previous section, we have presented the lifting process.
If, however, in one of the lifting steps  conditions of the lemma are not met, lifting fails and 
we now refine the original clause set in order to rule out the non-liftable conflicting core.
Again, since lifting fails at one of the approximation steps, we consider the different
approximation steps for refinement.

\subsubsection{Linear Approximation Refinement.}
A Linear transformation
enables further instantiations of the abstracted clause compared to the original, 
that is, two variables that were the same can now be instantiated differently. 
If the conflicting core of the approximation contains such instances the lifting fails.

\begin{defin}[Linear Approximation Refinement]
Let $N$ be a set of variable disjoint clauses, 
$N\apr^* N_k \Rightarrow_{\text{LI}} N_{k+1}$ and
$N^\bot_{k+1}$ be a conflicting core of $N_{k+1}$ where Lemma~\ref{liftlin} does not apply.
Let $C'\sigma=(\Gamma\{x \mapsto x'\},\Gamma \imp E[x']_q)\sigma \in N^\bot_{k+1}$
 such that $x\sigma$ and $ x'\sigma$ have no common instances. 
Let $C\in N$ be the Ancestor of $C'\in N_{k+1}$.
Then the \emph{linear approximation refinement} of $N$, $C$, $x$, $x'$, $\sigma$ is the clause
set $N\setminus\{C\}\cup\{C\tau_1,\ldots,C\tau_n\}$ where the $C\tau_i$ are the specific
instances of $C$ with respect to the substitutions $\{x\mapsto x\sigma\}$ and
$\{x\mapsto x'\sigma\}$.
\end{defin}

Note that if there is no $C'\sigma$, where $x\sigma$ and $ x'\sigma$ have no common instances, 
it implies that there is a substitution $\tau$
where Lemma~\ref{liftlin} applies on $N^\bot_{k+1}\tau$. Hence, $N^\bot_{k+1}\tau$ is a liftable
conflicting core.

Let $N_0\apr^* N_{k-1}= N\cup\{\Gamma \imp E[x]_{p,q}\} \Rightarrow_{\text{LI}} N_{k}= N\cup\{\Gamma\{x \mapsto x'\},\Gamma \imp E[x']_q\}$
and the core $N^\bot_k$ of $N_k$ contains the clause $C'\sigma=(\Gamma\{x \mapsto x'\},\Gamma \imp E[x']_q)\sigma$, where $x\sigma$ and $x'\sigma$ have no common instances.
After applying Linear Approximation Refinement, there are $C\tau_i$ and $ C\tau_j$ with $i\neq j$ such that  $C\tau_i$ contains all instances where $\{x \mapsto x\sigma \}$
and $C\tau_j$ contains all instances where $\{x \mapsto x'\sigma \}$. 
Assume there is a $C''$ with an ancestor $C\tau$ such that $C'\sigma$ is an instance of $C''$.
This would imply that $C\tau$ has instances, where $\{x \mapsto x\sigma\} $ and $\{x \mapsto x'\sigma\} $.
Then $C\tau_i= C\tau= C\tau_j$, which is a contradiction to Definition \ref{specInst}.

\begin{exmp}\label{LinAbsEx}
Let $N_0 = \{ P(x,x);$ $P(y,a),P(z,b) \imp \}$. 
Then $N_k=  \{ P(x,x');$ $P(y,a),P(z,b) \imp \}$ is a Linear transformation of $N_0 $ and
and $N^\bot_k=\{ P(a,a);$ $P(a,b);$ $P(a,a),P(a,b)\imp\}$ is a conflicting core of $N_k$.\\
Due to $P(a,b)= P(x,x')\{x\mapsto a, x' \mapsto b\}$ lifting fails.
The Linear Approximation Refinement replaces $P(x,x)$ in $N_0$ with $P(a,a)$ and $P(b,b)$.
In the refined approximation $N'_k = \{ P(a,a);P(b,b);$ $P(y,a),P(z,b)\imp\}$
the violating clause $P(a,b)$ is not an instance of $N'_k$ and hence, the not-liftable conflicting core $N^\bot_k$ cannot be found again.
\end{exmp} 

\subsubsection{Shallow Approximation Refinement.}
The Shallow transformation is somewhat more complex than linear transformation, 
but the idea behind it is very similar to the linear case.
As mentioned before, the Shallow transformation can always be lifted 
if the set of common variables
$\vars(\Gamma_2,s) \cap \vars(\Gamma_1,E[x]_p)$ is empty. 
Otherwise, each such variable potentially 
introduces instantiations that are not liftable.

\begin{defin}[Shallow Approximation Refinement]
Let $N$ be a set of variable disjoint clauses, 
$N\apr^* N_k \Rightarrow_{\text{SH}} N_{k+1}$ and
$N^\bot_{k+1}$ be a conflicting core of $N_{k+1}$ where Lemma~\ref{liftshallow} does not apply.
Let $C_R$ be the resolvent from the final Shallow rule application such that $C_R \neq (\Gamma \imp E[s]_{p})\sigma_R$ for any $\sigma_R$. 
Let $C_1\sigma_1 \in N^\bot_{k+1}$ and $C_2\sigma_2 \in N^\bot_{k+1}$ be the parent clauses of $C_R$.
Let $y \in \text{dom}(\sigma_1) \cap \text{dom}(\sigma_2)$, where $y\sigma_1$ and $ y\sigma_2$ have no common instances. 
Let $C\in N$ be the Ancestor of $C_1\in N_{k+1}$.
Then the \emph{shallow approximation refinement} of $N$, $C$, $x$, $\sigma_1$, $\sigma_2$ is the clause
set $N\setminus\{C\}\cup\{C\tau_1,\ldots,C\tau_n\}$ where the $C\tau_i$ are the specific
instances of $C$ with respect to the substitutions $\{x\mapsto x\sigma_1\}$ and
$\{x\mapsto x\sigma_2\}$.
\end{defin}

As in Linear Approximation Refinement, if for every resolvent $C_R\sigma$  $y\sigma_1$ and $ y\sigma_2$ have common instances, it implies that there is a substitution $\tau$
where Lemma~\ref{liftshallow} applies on $N^\bot_{k+1}\tau$.
After applying Shallow Approximation Refinement, there are $C\tau_i$ and $ C\tau_j$ with $i\neq j$ such that  $C\tau_i$ contains all instances where $\{x \mapsto x\sigma_1\} $
and $C\tau_j$ contains all instances where $\{x \mapsto x\sigma_2\}$. 
Hence, $C\tau_i$ is now the ancestor of $C_1\sigma_1$, while $C\tau_j$ is the ancestor of $C_2\sigma_2$.
Since they have different ancestors, they can no longer be resolved on their $S$-atoms which now have different predicates.
Hence $C_R$ is no longer a resolvent in the conflicting  core.

\begin{exmp}
Let $N_0 = \{ P(f(x,g(x))); P(f(a,g(b))\imp\}$ with signature $\Sigma = a/0,$ $b/0,$ $g/1, f/2$.
Then $N_k=  \{  S(z) \imp  P(f(x,z)); S(g(y));$ $ P(f(a,g(b))\imp\}$ is a Shallow transformation of $N_0 $ and
and $N^\bot_k=\{  S(g(b)) \imp  P(f(a,g(b))); S(g(b));$ $  P(f(a,g(b))\imp\}$ is a conflicting core of $N_k$.\\
The clauses $S(g(b)) \imp  P(f(a,g(b)))$ and $ S(g(b))$ have the resolvent $P(f(a,g(b)))$, which is not an instance of  $P(f(x,g(x)))$.
The Shallow Approximation Refinement replaces $P(f(x,g(x)))$ in $N_0$ with $P(f(a,g(a)))$, $P(f(b,g(b)))$,\\ $P(f(g(x),g(g(x))))$ and $P(f(f(x,y),g(f(x,y))))$.\\
The approximation of the refined $N_0$ is now satisfiable. 
\end{exmp}

\subsubsection{Horn Approximation Refinement.}
Lifting a core of a Horn transformation fails, if  the positive literals removed by 
the Horn transformation are not dealt with in the approximated proof.
Since Lemma \ref{lifthorn2} only handles cases where the approximated clause appears uninstantiated in the conflicting core,
the Horn Approximation Refinement is used to ensure such a core exists. 

\begin{defin}[Horn Approximation Refinement]
Let $N$ be a set of variable disjoint clauses, 
$N\apr^* N_k \Rightarrow_{\text{HO}} N_{k+1}$,
$N_k = N\cup\{\Gamma \imp E_1,E_2\}$ and $N_{k+1} = N\cup\{\Gamma \imp E_1\}$ and
$N^\bot_{k+1}$ be a conflicting core of $N_{k+1}$ where Lemmas ~\ref{lifthorn1} and ~\ref{lifthorn2} do not apply.
Let $(\Gamma \imp E_1)\sigma \in N^\bot_{k+1}$ 
 be a clause from the final Horn rule application such that $\sigma$ is not a variable renaming and
$N_k^j\tau_j\not \models  (\Gamma \imp E_1,E_2)\sigma'_j\rightarrow (\Gamma \imp E_1)\sigma$ for any $N_k^j\subseteq N_k$,$\tau_j$ and $\sigma'_j$. 
Let $C\in N$ be the Ancestor of $\Gamma \imp E_1\in N_{k+1}$
and $\sigma'$ a substitution such that $\sigma\sigma'$ is linear for $C$.
Then the \emph{horn approximation refinement I} of $N$, $C$, $\sigma$, $\sigma'$ is the clause
set $N\setminus\{C\}\cup\{C\sigma\sigma',C\tau_1,\ldots,C\tau_n\}$ where the $C\tau_i$ are the specific
instances of $C$ with respect to the substitutions $\sigma\sigma'$.
\end{defin}

Note that the condition for the extended version of  specific instantiation to have a finite representation is not generally met by an arbitrary $\sigma$.
Therefore, $\sigma$ may need to be further instantiated or even made ground. 
After the Horn Approximation Refinement, Lemma~\ref{lifthorn2} can be applied on 
the clause with ancestor $C\sigma\sigma'$.

\begin{exmp}
Let $N_0 = \{ P(x),Q(x);  P(a)\imp\}$ with signature $\Sigma = a/0, f/1$. 
The Horn transformation $N_k=\{  P(x); P(a) \imp\}$ has a conflicting core 
$N^\bot_k= \{  P(a); P(a) \imp\}$.
We pick $ \imp P(a)$ as the instance of $ P(x) \in N^\bot_k$ to use for  the Horn Approximation Refinement.
The result is $N'_0=\{  P(a),Q(a);$ $ P(f(x)),Q(f(x));$ $  P(a)\imp\} $ 
and its approximation also has $N^\bot_k$ as a conflicting core.
However, now Lemma \ref{lifthorn2} applies.
\end{exmp}

\begin{lem}[Completeness]\label{complete}
Let $N$ be an unsatisfiable clause set and $N_k$ its approximation.
Then, there exists a conflicting core of $N_k$ that can be lifted to $N$.
\end{lem}

\begin{proof}
by induction on the number $k$ of approximation steps. The case $k=0$ is obvious.
For $k>0$, let $N\apr^* N_{k-1} \apr N_k$.
By the inductive hypothesis, there is a conflicting core  $N^\bot_{k-1}$ of $N_{k-1}$ which can 
be lifted to $N$.\\
The final approximation rule application is either a Linear, a Shallow, a Horn or a Monadic transformation,
considered below by case analysis.

Linear Case. Let $N\apr^* N_{k-1}= N'\cup\{\Gamma \imp E[x]_{p,q}\} \Rightarrow_{\text{LI}} N_k= N'\cup\{\Gamma\{x \mapsto x'\},\Gamma \imp E[x']_q\}$.
For every $(\Gamma \imp E[x]_{p,q})\sigma_j \in N^\bot_{k-1}$ $1\leq j\leq m$, $(\Gamma \imp E[x]_{p,q})\sigma_j \models (\Gamma\{x \mapsto x'\},\Gamma \imp E[x']_q)(\{x' \mapsto x\}\sigma_j)$. 
Hence $N^\bot_k=N^\bot_{k-1}\setminus\{(\Gamma \imp E[x]_{p,q})\sigma_j\mid 1\leq j\leq m\}\cup \{(\Gamma\{x \mapsto x'\},\Gamma \imp E[x']_q)\{x' \mapsto x\}\sigma_j \mid 1\leq j\leq m\}$ 
is a conflicting core of $N_k$.
By Lemma \ref{liftlin} $N^\bot_k$ can be lifted back to $N^\bot_{k-1}$.
Hence, the conflicting core $N^\bot_k$ can be lifted to $N$.

Shallow Case. Let $N\apr^*  N_{k-1}=N'\cup\{\Gamma \imp E[s]_{p}\} \Rightarrow_{\text{SH}} N_{k}= N'\cup\{S(x),\Gamma_1 \imp E[x]_{p}\}\cup\{ \Gamma_2 \imp S(s)\}$.
We construct $N^\bot_S$ from $N^\bot_{k-1}$ by replacing every $(\Gamma \imp E[s]_{p})\sigma_j \in N^\bot_{k-1}$ $1 \leq j \leq m$ with $(S_j(x),\Gamma_1 \imp E[x]_p)\sigma_j$ and $(\Gamma_2 \imp S_j(s))\sigma_j$.
$N^\bot_S$ is a conflicting core, which by $m$ applications of Lemma \ref{liftshallow} on each $S_j$ can be lifted to $N^\bot_{k-1}$.
From $N^\bot_S$ we get $N^\bot_{k}$ by renaming every $S_j$ into $S$, which is a conflicting core of $N_k$.
The existence of  $N^\bot_S$ shows that $N^\bot_{k}$ can be lifted to $N^\bot_{k-1}$.

Horn Case. W.l.o.g. let $N\apr^*  N_{k-1}=N'\cup\{\Gamma \imp E_1,E_2\} \Rightarrow_{\text{HO}} N_{k}= N'\cup\{\Gamma \imp E_1\}$.
Let $C = \Gamma \imp E_1 ,E_2$ and $C' = \Gamma \imp E_1$.
If $C\sigma \in N^\bot_{k-1}$ holds for at most one $\sigma$, we construct $N^\bot_{k}$ from $N^\bot_{k-1}$ by replacing $C\sigma$ with $C'\sigma$ such that $N^\bot_{k}\subseteq N_k$.
Since $C'\sigma$ subsumes $C\sigma$, $N^\bot_{n} \sat N^\bot_{n}\cup\{C\sigma\}$.
As $N^\bot_{k}\cup\{C\sigma\}$ is a superset of $N^\bot_{k-1}$, $N^\bot_{k}$ is therefore a ground conflicting core of $N_k$.
If $C'\sigma$ and $C\sigma$ are already equivalent,  $N^\bot_{k}$ can be lifted to $N^\bot_{k-1}$.
Otherwise, let $N'^\bot_{k-1}$ be $N^\bot_{k-1}$ where  $C\sigma$  is instead replaced by $E_2\sigma$.
Again since $E_2\sigma$ subsumes $C\sigma$, $N'^\bot_{k-1}$ is a ground conflicting core. 
As shown before, $(N'^\bot_{k-1}\setminus\{E_2\sigma\})\cup (N^\bot_{k}\setminus\{C'\sigma\})=N^\bot_{k-1}$ is a lifting from $N_k$ to $N_{k-1}$.\\
Assume $C\sigma_1 \in N^\bot_{k-1}$ and $C\sigma_2 \in N^\bot_{k-1}$ holds for $\sigma_1\neq \sigma_2$. 
In this case the original clause $C$ can be specifically instantiated in such a way that $C\sigma_1$ and $C\sigma_2$ are no longer instances of the same clause,
while $N^\bot_{k-1}$ remains a conflicting core.
Hence, after finitely many such partitions eventually the first case will hold. 

Monadic Case. Let $N\apr^* N_{k-j} \Rightarrow^*_{\text{MO}} N_k$ where $N_{k-j}$ has no occurrence of an atom 
$T(f_P(t_1,\dots,t_n))$ and $N_{k}$ no occurrence of an atom $P(t_1,\dots,t_n)$ and
all introduced atoms in the transformation are of the form $T(f_P(s_1,\dots,s_n))$.
By the inductive hypothesis, there is a ground conflicting core  $N^\bot_{k-j}$ of $N_{k-j}$ which can be lifted to $N$.
By Lemma \ref{Apr_sound}(ii) Monadic transformation  preserves unsatisfiability and therefore $\Proj{P}(N^\bot_{k-j})$ is a ground conflicting core of $N_{k}$. 
$\Proj{P}(N^\bot_{k-j})$ can be lifted to $\R{P}(\Proj{P}(N^\bot_{k-j}))=N^\bot_{k-j}$ a conflicting core of $N_{k-j}$.
\end{proof}

The above lemma considers static completeness, i.e., it does not tell how the conflicting
core that can eventually be lifted is found. One way is to enumerate all refutations of  $N_k$
in a fair way. A straightforward fairness criterion is to enumerate the refutations by increasing
term depth of the clauses used in the refutation. Since the decision procedure on the
mslH fragment~\cite{Weidenbach99cade} generates only finitely many different
non-redundant clauses not exceeding a concrete term depth with respect to the renaming of variables, eventually the 
liftable refutation will be generated.

\section{Future and Related Work} \label{sec:furewo}

The condition for the lifing lemma for Shallow transformation (Lemma \ref{liftshallow})  is stronger than necessary, as the following example shows.  

\begin{exmp}\label{shallowex2}
Let $N_0 = \{ P(x,z), Q(y,z) \imp R(x,f(y)); P(a,a);$ $ P(a,b); $ \\ $ Q(b,a), Q(b,b);$ $ R(a,f(b)) \imp\}$  and
 $N_k = \{ S(y),P(x,z) \imp R(x,y); Q(y,z) \imp S(f(y)); P(a,a); P(a,b); Q(b,a) , Q(b,b);$ $  R(a,f(b))  \imp\}$ is a Shallow transformation of $N_k$. 
$N_0$ and $N_k$ are unsatisfiable and $N^\bot_k=\{ S(f(b)),$ $P(a,a) \imp R(a,f(b)); Q(b,a) \imp S(f(b)); S(f(b)),P(a,b) \imp R(a,f(b)); Q(b,b) \imp S(f(b));$ $ P(a,a);  P(a,b);  Q(b,a), Q(b,b);  R(a,f(b))\imp\}$
is a conflicting core of $N_k$. Lifting $N^\bot_k$ fails because the resolvent $P(a,a), Q(b,b) \imp R(a,f(b))$ is not an instance of  $P(x,z), Q(y,z) \imp R(x,f(y))$.
However, if we ignored the violating resolvents, it would result in the valid conflicting core $N^\bot=\{ P(a,a),Q(b,a) \imp R(a,f(b)); P(a,b),Q(b,b) \imp R(a,f(b)); P(a,a);  P(a,b);$ $ Q(b,a), Q(b,b);$\\ $ R(a,f(b)) \imp\}$.
\end{exmp} 

This does not break lifting. 
The shallow refinement will partition the clause in such a way that the resolvents that violate the lifting condition are one-by-one removed.
In Example~\ref{shallowex2}, the refinement would partition  $P(x,z), Q(y,z) \imp R(x,f(y))$ on the variable $z$.
This will result in $S(f(b)),$ $P(a,a) \imp R(a,f(b))$ and  $Q(b,b) \imp S(f(b))$ containing different $S$-predicates and hence no longer being resolvable. 

However, a refinement is not necessary to achieve this effect. 
The necessary information can be taken from the refutation 
and incorporated into the conflicting core during construction.

If a problem $N$ is unsatisfiable, not only does there exist an unsatisfiability proof but one where $S$-literals only occur on leaves.
Such a proof can be found by a ordered resolution calculus through selecting negative $S$-literals and an ordering where positive $S$-literals are strictly maximal.
Given such a setting a solver will only resolve a clause  
$S(x),\Gamma_1 \imp E[x]_{p_1,\dots,p_n}$ with $\Gamma_2 \imp S(s)$ on the $S$-atom 
and hence any $S$-atom will only appear at the leaves of the refutation.

In such a proof, we then uniquely rename the $S$-predicate in each pair of leaves.
The conflicting core constructed from this proof then only allows resolutions on $S$-literals that also occur in the proof. 
On this core we can then check the lifting condition.

In example \ref{shallowex2} the core would then instead be 
$\{ S_1(f(b)),$ $P(a,a) \imp R(a,f(b));$ $ Q(b,a) \imp S_1(f(b)); S_2(f(b)),P(a,b) \imp R(a,f(b)); Q(b,b) \imp S_2(f(b));$ $ P(a,a); $ $ P(a,b);  Q(b,a), Q(b,b);  R(a,f(b))\imp\}$.
This core is liftable to $N^\bot$ by Lemma \ref{liftshallow}.

\subsubsection{Related Work}
In ''A theory of abstractions''~\cite{Giunchiglia:1992:TA:146945.146951}  Giunchiglia and Walsh don't define an actual approximation but a general framework to classify and compare approximations, which are here called abstractions. 
They informally define abstractions as ''the process of mapping a representations of a problem'' that ''helps deal with the problem in the original search space by preserving certain desirable properties`` and ''is simpler to handle``. 

In their framework an abstraction is a mapping between formal systems, i.e., a triple of a language, axioms and deduction rules, which satisfy one of the following conditions:
An increasing abstraction (TI) $f$ maps theorems only to theorems, i.e., if $\alpha$ is a theorem, then $f(\alpha)$ is also a theorem, 
while a decreasing abstraction (TD)   maps only theorems to theorems, i.e., if $f(\alpha)$ is a theorem, then $\alpha$ was also a theorem.

Furthermore, they define dual definitions for refutations, where not theorems but formulas that make a formal system inconsistent  are considered.
An increasing abstraction (NTI) then maps inconsistent formulas only to inconsistent formulas and vice versa for decreasing abstractions (NTD).

They list several examples of abstractions such as  ABSTRIPS by Sacerdoti~\cite{Sacerdott:1973:PHA:1624775.1624826}, a GPS planning method by Newell and Simon~\cite{Newell:1972:HPS:1095704}, Plaisted's theory of abstractions~\cite{journals/ai/Plaisted81}, propositional abstractions exemplified by Giunchiglia~\cite{conf/ecai/GiunchigliaG88}, predicate abstractions by by Plaisted~\cite{journals/ai/Plaisted81} and Tenenberg~\cite{Tenenberg87preservingconsistency}, domain abstractions by Hobbs~\cite{Hobbs85granularity} and Iemielinski~\cite{Imielinski:1987:DAL:1625995.1626083} and ground abstractions introduced by Plaisted~\cite{journals/ai/Plaisted81}.

With respect to their notions the approximation described in this paper is an abstraction where the desirable property is the over-approximation and the decidability of the fragment makes it simpler to handle. More specifically in the context of \cite{Giunchiglia:1992:TA:146945.146951} the approximation is an NTI abstraction for refutation systems, i.e., it is an abstraction that preserves inconsistency of the original. 

In  Plaisted~\cite{journals/ai/Plaisted81} three classes of abstractions are defined. The first two are ordinary and weak abstractions, which share the condition that if $C$ subsumes $D$ then every abstraction of $D$ is subsumed by some abstraction of $C$. However, our approximation falls in neither class as it violates this condition via the Horn approximation. For example $Q$ subsumes $P, Q$, but the Horn 
approximation $P$ of $P, Q$ is not subsumed by any approximation of $Q$.    
The third class are generalization functions, which change not the problem but abstract the resolution rule of inference. 

The theorem prover iProver uses the Inst-Gen~\cite{DBLP:conf/birthday/Korovin13} method, 
where a first-order problem is abstracted with a SAT problem by replacing every variable by the fresh constant $\bot$. 
The approximation is solved by a SAT solver and its answer is lifted to the original 
by equating abstracted terms with the set they represent, e.g., if $P(\bot)$ is true in a
 model returned by the SAT solver, then all instantiations of the original $P(x)$ are 
considered true as well. 
Inst-Gen abstracts using an  under-approximation of the original clause set.
In case the lifting of the satisfying model is inconsistent, 
the clash is resolved by appropriately instantiating the involved clauses, which mimics an inference step.
This is the dual of our method with the roles of satisfiability and unsatisfiability switched.
A further difference, however, is that Inst-Gen only finds finite models after approximation, 
while our approximation also discovers infinite models.
For example the simple problem $\{ P(a)$, $\neg P(f(a))$, $P(x) \imp P(f(f(x)))$, $P(f(f(x))) \imp P(x)\}$ has the satisfying model where $P$ is the set of even numbers.
However, iProver's approximation can never return such a model as any $P(f^n(\bot))$ will necessarily abstract both true and false atoms and therefore instantiate new clauses infinitely. 
Our method on the other hand will produce the approximation $\{P(a)$, $\neg P(f(a))$, $S(y) \imp P(f(y))$,  $P(x) \imp S(f(x))$,  $P(f(f(x))) \imp P(x)\}$, which  is saturated after inferring $P(x) \imp P(f(f(x)))$ and $\neg S(f(a))$.

In summary, we have presented the first sound and complete calculus for
first-order logic based on an over-approximation-refinement loop. There
is no implementation so far, but the calculus will be practically useful if a problem
is close to the mslH fragment in the sense that only a few refinement loops
are needed for finding the model or a liftable refutation. 
The abstraction relation is already implemented and applying it to all satisfiable non-equality
problems TPTP version~6.1 results in a success rate of 34\%, i.e., for all these problems
the approximation is not too crude and directly delivers the result.

It might be possible
to apply our idea to other decidable fragments of first-order logic.
However, then they have to support via approximation the presented lifting and refinement
principle. 

Our result is also a first step towards a  model-based guidance of first-order
reasoning. We proved that a model of the approximated clause set is also a model for the
original clause set. For model guidance, we need this property also for partial models. For example,
in the sense that if a clause is false with respect to a partial model operator on
the original clause set, it is also false with respect to a partial model operator on
the approximated clause set. This property does not hold for the standard superposition partial
model operator and the mslH approximation suggested in this paper. It is subject to future research.

\bibliographystyle{plain}
\bibliography{abstractions}

\clearpage
\appendix

\ \\
\ \\
\ \\
\ \\
\ \\
\ \\

\section{Skeleton and Partial Minimal Model Construction}

As mentioned before, Lemma \ref{termSkelLem} also holds if satisfiability of $N'$ is dropped and $I$ is replaced
by the superposition partial minimal model operator~\cite{Weidenbach01handbook}.

\begin{defin}[Partial Minimal Model Construction]
Given the set of ground clauses $N_g$ of $N$ and an ordering $\prec$ we construct an interpretation $\I{N}$ for $N$, called a partial model,  inductively as follows:
\begin{align*}
 \I{C} &:= \bigcup_{ D \in N_g,D\prec C} \delta_D\\
 \delta_D &:= \left\{
  \begin{array}{l l}
    \{P\} & \quad \text{if $D=D'\vee P$, P strictly maximal and $\I{D}\unsat D$}\\
    \emptyset & \quad \text{otherwise}
  \end{array} \right. \\
\I{N} &:= \bigcup_{C \in N_g} \delta_C
\end{align*}
Clauses $C$ with $\delta_C \neq \emptyset$ are called productive.
\end{defin}

Note that this construction doesn't terminate since the ground clause set of $N$ is generally infinite.

%

\begin{lem}\label{termSkelLem3}
Let $N_k$ be a monadic clause set and $N_0$ be its approximation via  $\apr$. 
If $P(s) \in \I{N_0}$  and $P$ is a predicate in $N_k$, then there exists a clause 
$C = \Gamma \imp \Delta,P(t) \in N_k$ and a substitution $\sigma$ such that
$s = \skt(t)\sigma$ and for each variable $x$ and predicate $S$ with $C = S(x),\Gamma' \imp \Delta,P(t[x]_p)$, $S(s'') \in \I{N_0}$, where $s=s[s'']_p$.
\end{lem}

\begin{proof}
By induction on $k$.\\
The base $N_k=N_0$ holds by definition of the model operator $\I{}$. \\
Let  $N = N_k\apr N_{k-1} \apr^* N_0$, $P(s) \in  \I{N_0}$ and $P$ is a predicate in $N_k$ and hence also in $N_{k-1}$. 
By the induction hypothesis, there exist a clause $C = \Gamma \imp \Delta,P(t) \in N_{k-1}$ and a substitution $\sigma$ such that
$s = \skt(t)\sigma$ and for each variable $x$ and predicate $S$ with $C = S(x),\Gamma' \imp \Delta,P(t[x]_p)$, $S(s'') \in \I{N_0}$, where $s=s[s'']_p$.

Let $\apr$ be a Horn transformation that replaces $\Gamma'' \imp \Delta',Q(t')$ with $\Gamma'' \imp Q(t')$.
If $C \neq \Gamma'' \imp Q(t')$, then $C \in N_k$ fulfills the claim.
Otherwise, $\Gamma'' \imp \Delta',Q(t) \in N_k$ fulfills the claim since $P=Q$ and $\Gamma' = \Gamma''$.

Let $\apr$ be a linear transformation that replaces $C_k = \Gamma'' \imp E[x]_{p,q}$ with $C_{k-1}=\Gamma'',\Gamma''\{x \mapsto x'\} \imp E[x']_{q}$.
If $C \neq C_{k-1}$, then $C \in N_k$ fulfills the claim.
Otherwise, $C_k=\Gamma'' \imp P(t)\{x' \mapsto x\} \in N_k$ fulfills the claim since $s = \skt(t)\sigma = \skt(t\{x' \mapsto x\})\sigma$ and $\Gamma'' \subseteq \Gamma'',\Gamma''\{x \mapsto x'\}$.

Let $\apr$ be a shallow transformation that replaces $C_k = \Gamma'' \imp E[s']_{p}$ with $C_{k-1}=S(x),\Gamma_1 \imp E[x]_p$ and $C'_{k-1}=\Gamma_2 \imp S(s')$.
Since $S$ is fresh, $C \neq C'_{k-1}$.
If $C \neq C_{k-1}$, then $C \in N_k$ fulfills the claim.
Otherwise, $C = C_{k-1} = S(x),\Gamma_1 \imp P(t[x]_p)$ and hence, $s = \skt(t[x]_p)\sigma$ and $S(s'')\in \I{N_0}$ for $s=s[s'']_p$.
Then by the induction hypothesis, there exist a clause $C_S = \Gamma_S \imp \Delta_S,S(t_S) \in N_{k-1}$ and a substitution $\sigma_S$ such that
$s'' = \skt(t_S)\sigma_S$ and for each variable $x$ and predicate $S'$ with $C_S = S'(x),\Gamma'_S \imp \Delta_S,P(t_S[x]_q)$, $S'(s''') \in \I{N_0}$, where $s''=s''[s''']_q$.
By construction, $C_S = C'_{k-1}$. 
Thus,  $s''= \skt(s')\sigma_S$ and $s = \skt(t[x]_p)\sigma$ imply there exists a $\sigma''$ such that $s =  \skt(t[s']_p)\sigma"$.
Furthermore since $\Gamma_1 \cup \Gamma_2 = \Gamma''$, if $C_k= S'(x),\Gamma''' \imp P(t[s']_{p})[x]_q$, then 
either $S'(x) \in \Gamma_1$ and thus  $S'(s'''') \in \I{N_0}$, where $s=s[s'''']_q$, or 
$S'(x) \in \Gamma_2$ and thus  $S'(s'''') \in \I{N_0}$, where $s[s'']_p=(s[s'']_p)[s'''']_q$.
Hence, $C_k \in N_k$  fulfills the claim.
\end{proof}

\begin{lem}\label{termSkelLem4}
Let $N$ be a clause set and $N'$ be its approximation via  $\apr$. 
If $P(s) \in \I{N'}$ $(T(f_p(s_1,\ldots,s_n))\in \I{N'})$ 
and $P$ is a predicate in $N$, then there exist a clause 
$\Gamma \imp \Delta,P(t) \in N$  ($\Gamma \imp \Delta,P(t_1,\ldots,t_n) \in N$)
and a substitution $\sigma$ such that
$s = \skt(t)\sigma$ ($s_i = \skt(t_i)\sigma$ for all $i$).
\end{lem}

\begin{proof}
Let $P_1,\ldots,P_n$ be the non-monadic predicates in $N$ and $N_{MO}=$ $\Proj{P_1}(\ldots(\Proj{P_n}(N)))$.
Then, $N_{MO}$ is monadic and also has $N'$ as its approximation via $\apr$.

 Let $P(s) \in \I{N'}$ and $P$ is a predicate in $N$. Since $P$ is monadic, $P$ is a predicate in $N_{MO}$.
Hence by Lemma ~\ref{termSkelLem3}, there exists a clause $\Gamma \imp \Delta,P(t) \in N_{MO}$ 
and a substitution $\sigma$ such that $s = \skt(t)\sigma$. 
Then, $\R{P_1}(\ldots(\R{P_n}(\Gamma \imp \Delta,P(t)))\ldots)= \R{P_1}(\ldots(\R{P_n}(\Gamma)\ldots) \imp \R{P_1}(\ldots(\R{P_n}(\Delta)\ldots),P(t)\in N$ fulfills the claim.

Let $T(f_p(s_1,\ldots,s_n))\in \I{N'}$ and $P$ is a predicate in $N$. $T$ is monadic and a predicate in $N_{MO}$.
Hence by Lemma ~\ref{termSkelLem3}, there exists a clause $\Gamma \imp \Delta,T(t) \in N_{MO}$ 
and a substitution $\sigma$ such that $f_p(s_1,\ldots,s_n) = \skt(t)\sigma$.
Therefore, $t=f_p(t_1,\ldots,t_n)$ with  $s_i = \skt(t_i)\sigma$ for all $i$.
Then, $\R{P_1}(\ldots(\R{P_n}(\Gamma \imp \Delta,T(f_p(t_1,\ldots,t_n))))\ldots)=\R{P_1}(\ldots(\R{P_n}(\Gamma)\ldots) \imp \R{P_1}(\ldots(\R{P_n}(\Delta)\ldots),P(t_1,\ldots,t_n) \in N$ fulfills the claim.
\end{proof}

%
%
%

\end{document}